\documentclass[1p,final]{elsarticle}
\usepackage{amsfonts,natbib,pslatex}

\usepackage{latexsym,amsmath,amsthm,amssymb,amsxtra,mathrsfs,times,CJK}

\usepackage{color}
\usepackage{amsfonts,color}
\usepackage{amssymb,amsmath,latexsym}
\usepackage{amssymb}

\usepackage[para]{threeparttable}

\newcommand{\SQ}{{\mathrm{SQ}}}
\newcommand{\NQ}{{\mathrm{NSQ}}}

\newcommand{\tr}{{\rm Tr}}
\newcommand{\gf}{{\mathbb{GF}}}

\newcommand{\C}{{\mathcal C}}

\newtheorem{theorem}{Theorem}[section]
\newtheorem{lemma}[theorem]{Lemma}

\newtheorem{example}[theorem]{Example}

\journal{Finite Fields and Their Applications}

\begin{document}

\begin{frontmatter}



\title{A class of optimal ternary cyclic codes and their duals \tnoteref{fn1}}
\tnotetext[fn1]{C. Fan's research was supported by
the Natural Science Foundation of China, Proj. No. 11571285, Z. Zhou's research was supported by
the Natural Science Foundation of China, Proj. No. 61201243, the Sichuan Provincial Youth Science and Technology Fund under Grant
2015JQO004, and the Open Research Fund of National Mobile Communications Research Laboratory, Southeast University under Grant 2013D10.}


\author[SWJTU]{Cuiling Fan\corref{cor1}}
 \ead{cuilingfan@163.com}
\author[UB]{Nian Li}
 \ead{nianli.2010@gmail.com}
\author[SWJTU]{Zhengchun Zhou}
 \ead{zzc@home.swjtu.edu.cn,zczhou@126.com}
 \cortext[cor1]{Corresponding author}
 \address[SWJTU]{School of Mathematics, Southwest Jiaotong University, Chengdu, 610031, China}
\address[UB]{Department of Informatics, University of Bergen, N-5020 Bergen, Norway}
\begin{abstract}
Cyclic codes are a subclass of linear codes and have applications in consumer electronics,
data storage systems, and communication systems as they have efficient encoding and
decoding algorithms. Let $m=2\ell+1$ for an integer $\ell\geq 1$ and $\pi$ be a generator of $\gf(3^m)^*$.
In this paper, a class of cyclic codes $\C_{(u,v)}$ over $\gf(3)$
with two nonzeros $\pi^{u}$ and $\pi^{v}$ is studied, where $u=(3^m+1)/2$, and $v=2\cdot 3^{\ell}+1$
is the ternary Welch-type exponent.
Based on a result on the non-existence of solutions to certain equation  over $\gf(3^m)$,
the cyclic code $\C_{(u,v)}$ is shown to have minimal distance four, which is the best minimal distance for any linear code over $\gf(3)$
with length $3^m-1$ and dimension $3^m-1-2m$
according to the  Sphere Packing bound.
The duals of this class of cyclic codes are also studied.
\end{abstract}


\end{frontmatter}

\begin{abstract}
Cyclic codes are a subclass of linear codes and have applications in consumer electronics,
data storage systems, and communication systems as they have efficient encoding and
decoding algorithms. In this paper, a class of three-weight cyclic codes over $\gf_{p}$
whose duals have two zeros is presented, where $p$ is an odd prime. The weight distribution
of this class of cyclic codes is settled. Some of the cyclic codes are optimal. The duals of a
subclass of the cyclic  codes are also studied and proved to be optimal.
\end{abstract}


\section{Introduction}

Let $p$ be a prime. An  $[n,k,d]$ linear code $\C$ over the finite field  $\gf(p)$ is a
$k$-dimensional subspace of $\gf(p)^n$ with minimum (Hamming)
distance $d$, and is called  cyclic if any cyclic shift of a codeword
is another codeword of $\C$. By identifying  $(c_0,c_1,\cdots,c_{n-1})\in \mathcal{C}$ with
 \[c_0+c_1x+c_2x^2+\cdots+c_{n-1}x^{n-1}\in\gf(p)[x]/(x^n-1),\]
 any cyclic code of length $n$ over $\gf(p)$ corresponds to an ideal of the polynomial residue class ring $\gf(p)[x]/(x^n-1)$.
 Note that every ideal of $\gf(p)[x]/(x^n-1)$ is principal. Thus, any cyclic code $\C$ can be expressed as $\C=\langle g(x) \rangle$, where $g(x)$ is monic and has the least degree. The polynomial  $g(x)$ is called the {generator polynomial} and $h(x)=(x^n-1)/g(x)$ is referred to as the { parity-check polynomial} of $\C$. The cyclic code $\mathcal{C}=\langle g(x) \rangle$ is said to have $t$ {nonzeros} if its parity-check polynomial $h(x)$ can be factorized as a product of $t$ distinct irreducible polynomials over $\gf(p)$ and accordingly the dual code $\C^{\perp}$ of $\C$ is said to have $t$ {zeros}.

Let $A_i$ denote the number of codewords with Hamming weight $i$ in a code $\C$ of length $n$ for $1\leq i\leq n$. The {weight enumerator} of $\C$ is defined by
 \[1+A_1z+A_2z^2+\cdots+A_nz^n,\]
 and the vector $(1, A_1, A_2, \ldots, A_n)$ is called the weight distribution of the code $\C$.  If $\C$ is a linear code, then the weight distribution of $\C$ gives the minimum distance and the error correcting capability of $\C$.
A code $\C$ is
said to be a $t$-weight code if the number of nonzero $A_i$
in the sequence $(A_1, A_2, \ldots, A_n)$ is equal to $t$ \cite{Lint}.
\vspace{2mm}

Cyclic codes are a subclass of linear codes and have important applications in consumer electronics,
data storage systems, and communication systems as they have efficient encoding and
decoding algorithms compared with the  linear block codes \cite{Klove}.
They also have applications in cryptography \cite{CDY05,DW05} and sequence design \cite{DYT}.
During the past few decades, cyclic codes have received a lot of attention and much progress have been made
(see \cite{CDY05}, \cite{Ding12}-\cite{Ding-Heelseth}, \cite{TFeng}, \cite{Li-CJ}-\cite{LiN}, \cite{Schmidt}, \cite{YangJing}, \cite{ZhengDB}-\cite{ZhouDingTC}  and the references therein).

Let $\gf(3^m)$ be the finite field with $3^m$ elements. Let $\pi$ be a generator of $\gf(3^m)^*$ and $m_{i}(x)$ be
the minimal polynomial of $\pi^i$ over $\gf(3)$, where $\gf(3^m)^*=\gf(3^m)\setminus\{0\}$ and $0\leq i\leq 3^m-2$.
Let $\C_{(u,v)}$ be the cyclic code over $\gf(3)$  with generator polynomial $m_u(x)m_v(x)$, where $u,v$ are two integers such that
$\pi^u$ and $\pi^v$ are nonconjugate. When $u=1$ and $v$ is an integer such that $x^v$ is a perfect nonlinear monomial,
Carlet, Ding, and Yuan \cite{CDY05} proved that
the code $\C_{(1,v)}$ has parameters $[3^m-1,3^m-1-2m,4]$ which is optimal according to the Sphere Packing bound.
Later, Ding and Helleseth \cite{Ding-Heelseth} constructed several classes of
optimal ternary cyclic codes $\C_{(1,v)}$ with parameters $[3^m-1,3^m-1-2m,4]$ using some monomials $x^v$ over $\gf(3^m)$ including almost perfect nonlinear monomials. Zhou and Ding obtained a class of optimal ternary cyclic codes $\C_{(u,v)}$ with parameters $[3^m-1,3^m-1-2m,4]$ by
choosing $(u,v)=((3^m+1)/2, (3^k+1)/2)$ where $m$ is odd and $k$ is even. Recently, Li {\it et al.} settled an open problem
proposed by Ding and Helleseth in \cite{Ding-Heelseth} and obtained some classes of optimal ternary
cyclic codes with parameters $[3^m-1,3^m-1-2m,4]$ and $[3^m-1,3^m-2-2m,5]$.
The duals of the aforementioned optimal ternary cyclic codes are discussed in \cite{FL07}, \cite{Li-CL}, \cite{YCD}, \cite{ZD13}, \cite{ZhouDingTC}.

The objective of this paper is to study a class of  ternary cyclic code $\C_{(u,v)}$  with generator polynomial $m_u(x)m_v(x)$, where $u=(3^m+1)/2$, and $v=2\cdot 3^{\ell}+1$ is the ternary Welch-type exponent proposed by Dobbertin {\it et al.} \cite{Dobb} where they studied the cross-correlation
between an $m$-sequence and its $v$-decimated version.
Based on a result on the non-existence of solutions to certain equation  over $\gf(3^m)$,
the cyclic code $\C_{(u,v)}$ is shown to have minimal distance four, which is the best minimal distance for such class of codes over $\gf(3)$. This family of cyclic codes is shown to
have parameters $[3^m-1,3^m-1-2m,4]$ and  is thus optimal according to the Sphere Packing bound. The duals of this class of cyclic codes are also studied.

\section{An equation over $\gf(3^m)$}\label{Pre}
In this section, we study an equation over $\gf(3^m)$, where $m=2\ell+1$.
The result on the non-existence of solutions to this equation
will be used to determine the minimal distance of a class of cyclic codes  in the sequel.

\begin{lemma}\label{Lemma-key1}
Let $m=2\ell+1$, where $\ell$ is a positive integer. Then for any given $\epsilon \in \gf(3)^*$, the equation
\begin{eqnarray}\label{eqn-main1}
(x^{3^\ell}+\epsilon)(x^{3^{\ell}}-x)=1
\end{eqnarray}
has no solution in $\gf(3^m)^*$.
\end{lemma}
\begin{proof}
It is clear that (\ref{eqn-main1}) holds if and only if
\begin{eqnarray}\label{eqn-sys-1}
\left\{
\begin{array}{llll}
x^{3^\ell}+\epsilon&=\theta\\
x^{3^\ell}-x&=\frac{1}{\theta}
\end{array} \right.\ \
\end{eqnarray}
holds for some $\theta\in \gf(3^m)^*$.
It is thus sufficient to prove that  (\ref{eqn-sys-1}) does not have solutions in $\gf(3^m)^*$ for any $\theta\in \gf(3^m)^*$.
Suppose on the contrary that (\ref{eqn-sys-1}) has a solution in $\gf(3^m)^*$ for some $\theta\in \gf(3^m)^*$, then $\theta$ should satisfy
\begin{eqnarray}\label{eqn-solution1}
\theta^{3^{\ell}}-\frac{1}{\theta^{3^{\ell}}}=\theta.
\end{eqnarray}
By (\ref{eqn-solution1}), we immediately have
$$
\left(\frac{\theta^2-1}{\theta}\right)^{3^\ell}=\theta,
$$
which means that
$$
(\theta^2-1)^{3^\ell}=\theta^{1+3^\ell}.
$$
Thus $\theta^2-1$ is a square in $\gf({3^m})^*$. On the other hand, (\ref{eqn-solution1}) holds for some $\theta\in \gf(3^m)^*$ if and only if
$$
\theta^{3^\ell}(\theta^{3^\ell}-\theta)=1.
$$
Recall that $m=2\ell+1$. Taking $3$ and $3^{\ell+1}$ powers on both sides of the equation above,
\begin{eqnarray}\label{eqn-basic1}
\theta^{3^{\ell+1}}(\theta^{3^{\ell+1}}-\theta^3)=1,
\end{eqnarray}
\begin{eqnarray}\label{eqn-basic2}
\theta(\theta-\theta^{3^{\ell+1}})=1.
\end{eqnarray}
Using (\ref{eqn-basic1}) and (\ref{eqn-basic2}), we obtain
\begin{eqnarray}\label{eqn-basic3}
\frac{\theta^{3^{\ell+1}}(\theta^{3^{\ell+1}}-\theta^3)+1}{(\theta(\theta-\theta^{3^{\ell+1}})+1)^2}-1=1.
\end{eqnarray}
Note that
\begin{eqnarray*}
&&\theta^{3^{\ell+1}}(\theta^{3^{\ell+1}}-\theta^3)+1-(\theta(\theta-\theta^{3^{\ell+1}})+1)^2\\
&=&(\theta^{3^{\ell+1}})^2-\theta^3\theta^{3^{\ell+1}}+1-(\theta^2(\theta-\theta^{3^{\ell+1}})^2+2\theta(\theta-\theta^{3^{\ell+1}})+1)\\
&=&((\theta^{3^{\ell+1}})^2+2\theta\theta^{3^{\ell+1}}+\theta^2)-(\theta^2(\theta-\theta^{3^{\ell+1}})^2+\theta^3\theta^{3^{\ell+1}})\\
&=&(\theta+\theta^{3^{\ell+1}})^2-\theta^2(\theta+\theta^{3^{\ell+1}})^2\\
&=&(\theta+\theta^{3^{\ell+1}})^2(1-\theta^2).
\end{eqnarray*}
Then (\ref{eqn-basic3}) becomes
\begin{eqnarray*}
\left(\frac{\theta+\theta^{3^{\ell+1}}}{\theta(\theta-\theta^{3^{\ell+1}})+1}\right)^2(1-\theta^2)=1.
\end{eqnarray*}
Therefore, $1-\theta^2$ is  a square in $\gf(3^m)^*$. This is a contradiction since $\theta^2-1$ is a square in $\gf(3^m)^*$ and $-1$ is a nonsquare in $\gf(3^m)^*$ if $m$ is odd. The proof of the lemma is finished.
\end{proof}

\section{A class of optimal ternary cyclic codes and their duals}
In this section, suppose that $(u,v)=((3^m+1)/2, 2\cdot 3^{\ell}+1)$ for an odd integer $m=2\ell+1\geq 3$.
We shall study the properties of the cyclic code $\C_{(u,v)}$ with two nonzeros $\pi^{u}$ and $\pi^v$  and its dual.

\subsection{The parameters of cyclic code $\C_{(u,v)}$}

The length of the cyclic code $\C_{(u,v)}$ is $3^m-1$, and its dimension is determined by sizes of the cyclotomic cosets modulo  $3^m-1$  containing $u$ and $v$. For $0\leq j\leq 3^m-2$. The cyclotomic coset modulo  $3^m-1$  containing $j$ is defined as
\[C_j=\{j\cdot3^s({\rm mod}\, 3^m-1): s=0,1,2,\cdots,m-1\}.\]

\begin{theorem}\label{thm-wb-1}
Let $m=2\ell+1$ and $(u,v)=((3^m+1)/2, 2\cdot 3^{\ell}+1)$. Then the code $\C_{(u,v)}$ is an optimal ternary cyclic code with parameters  $[3^m-1,3^m-1-2m,4]$.
\end{theorem}

\begin{proof}
Note that $\gcd(u,3^m-1)=2$ and $\gcd(v,3^m-1)=1$ since $m$ is odd.
Then, it can be readily verified that $|C_u|=|C_v|=m$ and $C_u\cap C_v=\emptyset$. This implies that
the dimension of $\C_{(u,v)}$ is equal to $3^m-1-2m$.

We now prove that the minimal distance $d$ of $\C_{(u,v)}$ is equal to $4$.
To this end, we first prove $d\geq 4$.
It is clear that $d\geq 2$ since $c_i\pi^{iu}\neq 0$ for any $0\leq i\leq 3^m-2$ and $c_i\in \gf(3)^*$.
By the definition of $\C_{(u,v)}$, it has a codeword of Hamming weight $2$ if and only if there exist two
elements $c_1,c_2\in \gf(3)^*$ and two distinct integers $0\leq t_1<t_2\leq 3^m-2$ such that
\begin{eqnarray}\label{eqn_system_1}
\left\{
\begin{array}{lllll}
c_1\pi^{ut_1}&+&c_2\pi^{ut_2}&=&0\\
c_1\pi^{vt_1}&+&c_2\pi^{vt_2}&=&0.
\end{array} \right.\ \
\end{eqnarray}
Note that $\gcd(v,3^m-1)=1$. It follows from the second equation of
(\ref{eqn_system_1}) that $c_1=c_2$ and  $t_2=t_1+ {(3^m-1)/ 2}$ since $t_1\neq t_2$. Then the first equation becomes
$2c_1\pi^{ut_1}=0$, which is impossible. Thus  the
code $\C_{(u,v)}$ does not have a codeword of Hamming weight $2$. We continue the proof to show that $\C_{(u,v)}$ has no codewords of weight $3$. Otherwise,
there exist three elements $c_1,c_2,c_3$ in $\gf(3)^*$ and three distinct integers
$0\leq t_1<t_2<t_3\leq 3^m-2$ such that
 \begin{eqnarray}\label{eqn_system_2}
\left\{
\begin{array}{lllll}
c_1\pi^{ut_1}&+c_2\pi^{ut_2}&+c_3\pi^{ut_3} &=&0\\
c_1\pi^{vt_1}&+c_2\pi^{vt_2}&+c_3\pi^{vt_3} &=&0.
\end{array} \right.\ \
\end{eqnarray}
Let $x_i=\pi^{t_i}$ for $i=1,2,3$. Then $x_1,x_2,x_3\in \gf(3^m)^*$ and are distinct, and  (\ref{eqn_system_2}) becomes
 \begin{eqnarray}\label{eqn_system_21}
\left\{
\begin{array}{lllll}
c_1x_1^u&+c_2x_2^u&+c_3x_3^u &=&0\\
c_1x_1^v&+c_2x_2^v&+c_3x_3^v &=&0.
\end{array} \right.\ \
\end{eqnarray}
Let $y_1=x_1/x_3$ and $y_2=x_2/x_3$. It then follows from (\ref{eqn_system_21}) that
\begin{eqnarray}\label{eqn_system_22}
\left\{
\begin{array}{lllll}
c_1y_1^u&+c_2y_2^u&+c_3 &=&0\\
c_1y_1^v&+c_2y_2^v&+c_3 &=&0.
\end{array} \right.\ \
\end{eqnarray}
We only need to consider the solutions of (\ref{eqn_system_22}) for $y_1, y_2\in \gf(3^m)^*\setminus \{1\}$, since $x_1,x_2,x_3$ are pairwise distinct. Due to symmetry it is sufficient to consider the following two cases.

{\textit{Case A}}: when $c_1=c_2=c_3=1$: In this case, we have
 \begin{eqnarray}\label{eqn_system_3}
\left\{
\begin{array}{lllll}
y_1^u&+y_2^u&+1 &=&0\\
y_1^v&+y_2^v&+1 &=&0.
\end{array} \right.\ \
\end{eqnarray}
Recall that $u=(3^m+1)/2$. We have $y^u=y$ if $y$ is a square in $\gf(3^m)^*$ and otherwise $y^u=-y$.
We distinguish among the following four cases to  prove that (\ref{eqn_system_3}) cannot
hold for any  $y_1, y_2\in \gf(3^m)^*\setminus \{1\}$.

(1) $y_1,y_2$ are squares in $\gf(3^m)^*$.
In this subcase, (\ref{eqn_system_3}) becomes
 \begin{eqnarray*}
\left\{
\begin{array}{lll}
y_1+y_2+1&=&0\\
y_1^v+y_2^v+1 &=&0
\end{array} \right.\ \
\end{eqnarray*}
which leads to
\begin{eqnarray}\label{eqn-final-1}
(1+y_1)^v=1+y_1^v.
\end{eqnarray}
Notice that
$$
(1+y_1)^{v}=1+y_1^v+y_1^{2\cdot 3^\ell}-y_1^{3^{\ell}+1}-y_1^{3^{\ell}}+y_1.
$$
This together with (\ref{eqn-final-1}) yields
\begin{eqnarray*}
(y_1^{3^\ell}-y_1)(y_1^{3^{\ell}}-1)=0
\end{eqnarray*}
which implies that $y_1\in \gf(3)$ since $\gcd(m,\ell)=\gcd(2\ell+1,\ell)=1$.
Thus $y_1=-1$ since $y_1\neq 0$ and $y_1\neq 1$.
This is a contradiction with the assumption that $y_1$ is a square in $\gf(3^m)^*$.

(2) $y_1$ is a square in $\gf(3^m)^*$ and $y_2$ is a nonsquare in $\gf(3^m)^*$.
By a similar routine calculation as case (1), we arrive at
$$
(y_1^{3^\ell}-y_1)(y_1^{3^\ell}-1)=-(1+y_1)^{2\cdot {3^\ell+1}}.
$$
Set $z_1=1+y_1$. Then we have
$$
(z_1^{3^\ell}-z_1)(z_1^{3^\ell}+1)=-z_1^{2\cdot {3^\ell+1}}.
$$
Let $\tilde{z}_1=z_1^{-1}$. Then dividing by $z_1^{2\cdot {3^\ell+1}}$ on both sides of the equation above
gives
\begin{eqnarray}\label{eqn-final-2}
(\tilde{z}_1^{3^\ell}-\tilde{z}_1)(\tilde{z}_1^{3^\ell}+1)=1.
\end{eqnarray}
By Lemma \ref{Lemma-key1}, (\ref{eqn-final-2}) cannot
hold for any $\tilde{z}_1\in \gf(3^m)^*$.

(3) $y_1$ is a nonsquare in $\gf(3^m)^*$ and $y_2$ is a  nonsquare in $\gf(3^m)^*$.
This case is similar to case (2).

(4) $y_1$ and $y_2$ are nonsquares in $\gf(3^m)^*$.
With a similar routine calculation as case (1), we have
$$
(y_1^{3^\ell}+1)(y_1^{3^\ell}-y_1)=1.
$$
This equation has no solutions in $\gf(3^m)^*$
due to Lemma \ref{Lemma-key1}.

{\textit{Case B}}: when $c_1=c_2=1$ and $c_3=-1$. The proof of this case is similar to Case A. We omit the details here.

The discussion above shows that $d\geq 4$. On the other hand, according to the Sphere Packing bound (c.f., p. 48, \cite{Huff}), the minimal distance of any
linear code with length $3^m-1$ and dimension $3^m-1-2m$ should be less than or equal $4$.
Hence $d=4$. This completes the proof.
\end{proof}

\subsection{The weights of the dual of $\C_{(u,v)}$}
In this section, we shall determine all the possible Hamming weights of
the duals of $\C_{(u,v)}$.  Using Delsarte's
Theorem \cite{Delsarte}, the dual of $\C_{(u,v)}$  is given by
\begin{eqnarray*}
\C^\bot_{(u,v)}=\{{\bf{c}}{(a,b)}:  a,b \in \gf({3^m})^2\}
\end{eqnarray*}
where the codeword
\begin{eqnarray*}
{\bf{c}}(a,b)=\left(\tr(a \pi^{-ui}+b \pi^{-vi})\right)_{i=0}^{3^m-2}
\end{eqnarray*}
and $\tr$ denotes the absolute trace from $\gf(3^m)$ to $\gf(3)$.

\begin{theorem}
The weight of the codeword ${\bf{c}}(a,b)$in $\C^\bot_{(u,v)}$ is $0$ if $a=b=0$, and otherwise takes values from
$$
\{2\cdot 3^{m-1},~ 2\cdot 3^{m-1}\pm 2\cdot 3^{\ell}, ~2\cdot 3^{m-1}\pm 3^{\ell}\}.
$$
\end{theorem}
\begin{proof}
Let $\chi_1$ and $\chi$ denote the canonical additive character of $\gf(3)$ and $\gf(3^m)$, respectively.
In terms of exponential sums, the weight of the codeword ${\bf c}_{(a,b)}=(c_0,c_1,\cdots,c_{3^m-2})$ in $\C^\bot_{(u,v)}$
is given by
\begin{eqnarray}\label{eqn-hammingweight1}
{\textrm{WT}}({\bf c}{(a,b)})&=& \#\{0\leq i\leq 3^m-2: c_i \neq 0\}\nonumber\\
&=& 3^m-1-\frac{1}{3}\sum_{i=0}^{3^m-2}\sum_{y\in {\gf({3}})}\chi_1({y c_i})\nonumber\\
&=& 3^m-1-\frac{1}{3}\sum_{i=0}^{3^m-2}\sum_{y\in {\gf({3}})}\chi_1({y (\tr(a\pi^{-ui})+\tr(b\pi^{-vi}))})\nonumber\\
&=&3^m-1-\frac{1}{3}\sum_{x\in \gf(3^m)^*}\sum_{y\in {\gf({3}})}\chi(ayx^u+byx^v)\nonumber\\
&=&2\cdot 3^{m-1}-\frac{1}{3}\sum_{y\in \gf(3)^*}\sum_{x\in {\gf({3^m}})}\chi(ayx^u+byx^v).
\end{eqnarray}
Note that $x^u=x$ if $x$ is a square in $\gf(3^m)$ and otherwise $x^u=-x$. It then follows from (\ref{eqn-hammingweight1}) that
\begin{eqnarray*}
{\textrm{WT}}({\bf c}{(a,b)})=2\cdot 3^{m-1}-\frac{1}{3}\sum_{y\in \gf(3)^*}(1+\sum_{x\in \SQ}\chi(ayx+byx^v)+\sum_{x\in \NQ}\chi(-ayx+byx^v)),
\end{eqnarray*}
which leads to
\begin{eqnarray}\label{eqn-hammingweight2}
&& 3{\textrm{WT}}({\bf c}{(a,b)})-2\cdot 3^{m}\nonumber\\
&=&2+\sum_{x\in \SQ}\chi(ax+bx^v)+\sum_{x\in \SQ}\chi(-ax-bx^v)+\sum_{x\in \NQ}\chi(-ax+bx^v)+\sum_{x\in \NQ}\chi(ax-bx^v)\nonumber\\
&=&2+\sum_{x\in \SQ}\chi(ax+bx^v)+\sum_{x\in \NQ}\chi(ax+bx^v)+\sum_{x\in \NQ}\chi(-ax+bx^v)+\sum_{x\in \SQ}\chi(-ax+bx^v)\nonumber\\
&=&\sum_{x\in \gf(3^m)}\chi(ax+bx^v)+\sum_{x\in \gf(3^m)}\chi(-ax+bx^v),
\end{eqnarray}
where $\SQ$ and $\NQ$ denote the set of all squares and nonsquares in $\gf(3^m)^*$ respectively, and the third equality
follows from the fact that $-x$ runs though $\NQ$ or $\SQ$ as $x$ ranges over $\SQ$ or $\NQ$ respectively.

We distinguish among the following three cases to determine all the possible weights of the codeword ${\bf c}{(a,b)}$.

{\textit{Case A}}: when $a=b=0$: It is clear that ${\textrm{WT}}({\bf c}{(a,b)})=0$ in this case.

{\textit{Case B}}: when $a=0$ or $b=0$: In this case, by (\ref{eqn-hammingweight2}) and the fact that $x^v$ is
a permutation over $\gf(3^m)$, we immediate obtain ${\textrm{WT}}({\bf c}{(a,b)})=2\cdot 3^{m-1}$.

{\textit{Case C}}: when $a\neq 0$ and $b\neq 0$: It follows from (\ref{eqn-hammingweight2})
that
\begin{eqnarray}\label{eqn-hammingweight3}
{\textrm{WT}}({\bf c}{(a,b)})=2\cdot 3^{m-1}-\frac{1}{3}(\hat{f}_v(\lambda)+\hat{f}_v(-\lambda))
\end{eqnarray}
where
$$
\hat{f}_v(\lambda)=\sum_{x\in \gf(3^m)}\chi(x^v-\lambda x)
$$
is the fourier transform of the power function $x^v$ at the point $\lambda=ab^{-1}$. It has been proven in \cite{Dobb}
that $\hat{f}_v(\lambda)\in \{0, \pm 3^{\ell+1}\}$ for each $\lambda \in \gf(3^m)^*$. This together
with (\ref{eqn-hammingweight3}) implies that ${\textrm{WT}}({\bf c}{(a,b)})$ takes values from the set
$\{2\cdot 3^{m-1},~ 2\cdot 3^{m-1}\pm 2\cdot 3^{\ell}, ~2\cdot 3^{m-1}\pm 3^{\ell}\}.$

The discussion above finishes the proof of this theorem.
\end{proof}

\subsection{Examples}

The following are some examples for this class of cyclic codes and their duals, which are generated by a Magma program.

\begin{example}
Let $p=3$, $m=5$ and $\pi$ be a generator of $\gf(3^m)^*$ with minimal polynomial $x^5+2x+1$. Then the code $\C_{(u,v)}$ is an optimal ternary cyclic
code with  generator polynomial $x^{10} + x^9 + 2x^8 + 2x^6 + 2x^5 + x^3 + 2x + 2$ and parameters
$[242,232,4]$.
The dual of the code  $\C_{(u,v)}$ has the following weight enumerator:
\begin{eqnarray*}
1+2420 z^{144}+12100 z^{153}+34364 z^{162}+7744 z^{171}+2420 z^{180}.
\end{eqnarray*}
\end{example}

\begin{example}
Let $p=3$, $m=7$ and $\pi$ be a generator of $\gf(3^m)^*$ with minimal polynomial $x^7+2x^2+1$. Then the code $\C_{(u,v)}$ is an optimal ternary cyclic
code with generator polynomial $ x^{14} + 2x^{12} + x^{10} + x^9 + 2x^8 + 2x^7 + 2x^5 + x^3 +x^2 + x+ 2$ and parameters
$[2186,2172,4]$.
The dual of the code  $\C_{(u,v)}$ has the following weight enumerator:
\begin{eqnarray*}
1+153020 z^{1404}+1040536 z^{1431}+2513900 z^{1458}+922492 z^{1485}+153020 z^{1512}.
\end{eqnarray*}
\end{example}

\begin{example}
Let $p=3$, $m=9$ and $\pi$ be a generator of $\gf(3^m)^*$ with minimal polynomial $x^9+2x^3+2x^2+x+1$. Then the code $\C_{(u,v)}$  is an optimal ternary cyclic
code with generator polynomial $x^{18} + 2x^{17} + x^{16} + x^{14} + x^{12} + 2x^{11} + 2x^{10}+
    x^8 +  2x^6 + x^5+2x^4 + x^2 + x + 2$ and parameters
$[19682,19664,4]$.
The dual of the code  $\C_{(u,v)}$ has the following weight enumerator:
\begin{eqnarray*}
1+10628280 z^{12960}+88214724 z^{13041}+192922964 z^{13122}+85026240 z^{13203}+10628280 z^{13284}.
\end{eqnarray*}
\end{example}

\section{Concluding Remarks}

In this paper, a class of ternary cyclic codes and their duals were
studied. Based on a result on the non-existence of solutions to an equation over $\gf(3^m)$, it was shown that this class of cyclic codes has parameters $[3^m-1,3^m-1-2m,4]$ for an odd integer $m$ and thus is optimal with respect to certain bound on general linear code.
The duals of this class of cyclic codes were shown to have at most five nonzero weights.  It would be interesting
if the weight distribution of the duals could be completely determined. The reader is invited to address this adventure.

\section*{Acknowledgments}
The authors are very grateful to the reviewer and the
Associate Editor, Prof. James W.P. Hirschfeld, for their comments and
suggestions that improved the presentation and quality of this
paper. This work was finished when the authors visited the Hong Kong University of Science and Technology. The authors
are grateful to Professor Cunsheng Ding for bringing them together
in the summer of 2014.

\end{document}